\documentclass{llncs}
\usepackage{makeidx}  % allows for indexgeneration
\begin{document}
\frontmatter          % for the preliminaries
\pagestyle{headings}  % switches on printing of running heads

\mainmatter              % start of the contributions
\title{On Simplex Pivoting Rules\\
 and Complexity Theory}
\titlerunning{Hamiltonian Mechanics}  % abbreviated title (for running head)
%\author{Ilan Adler\inst{1} \and Christos Papadimitriou\inst{2}
%\and  Aviad Rubinstein\inst{2}
\author{Ilan Adler \and Christos Papadimitriou\thanks{The research of Christos Papadimitriou and Aviad Rubinstein is supported by NSF Grant CCF-0964033.}
\and  Aviad Rubinstein$^\star$
}
\authorrunning{Ilan Adler et al.} % abbreviated author list (for running head)
%Craig Chambers, Kim B. Bruce, and Elisa Bertino}
%
\institute{ University of California, Berkeley CA 94720, USA.\\}
%\email{I.Ekeland@princeton.edu},\\ WWW home page:
%\texttt{http://users/\homedir iekeland/web/welcome.html}
%\and
%EECS department, University of California, Berkeley CA 94720, USA,\\}

\maketitle              % typeset the title of the contribution

\begin{abstract}
 We show that there are simplex pivoting rules for which it is PSPACE-complete to tell if a particular basis will appear on the algorithm's path.
  Such rules cannot be the basis of a strongly polynomial algorithm, unless P = PSPACE.  We conjecture that the same can be shown for most known variants of
  the simplex method.
   However, we also point out that Dantzig's shadow vertex algorithm has a polynomial path problem.  Finally, we discuss in the same context randomized pivoting rules.
\end{abstract}
\newpage
\pagenumbering{arabic}
\section{Introduction}
Linear programming was famously solved in the late 1940s by Dantzig's simplex method  \cite{GBD_LP}; however, many variants of the simplex method were eventually proved to have exponential worst-case performance \cite{KM}, while, around the same time, Karp's 1972 paper on NP-completeness \cite{RK} mentions linear programming as a rare problem in NP which resists classification as either NP-complete or polynomial-time solvable.  Khachiyan's ellipsoid algorithm \cite{LK} resolved positively this open question in 1979, but was broadly perceived as a poor competitor to the simplex method.  Not long after that, Karmarkar's interior point algorithm \cite{NK} provided a practically viable polynomial alternative to the simplex method. However, there was still a sense of dissatisfaction in the community:   The number of iterations of both the ellipsoid algorithm and the interior point method depend not just on the dimensions of the problem (the number of variables $d$ and the number of inequalities $n$) but also on the number of bits needed to represent the numbers in the input; such algorithms are sometimes called ``weakly polynomial''.

A {\em strongly polynomial algorithm}
%\footnote{Often termed ``real good algorithm'' in the optimization community.}
for linear programming (or any problem whose input is an array of integers) is one that is a polynomial-time algorithm in the ordinary sense (always stops within a number of steps that is polynomial in the total number of bits in the input), but it also takes a number of elementary arithmetic operations that is polynomial in the dimension of the input array.   Strongly polynomial algorithms exist for many network-related special cases of linear programming, as was first shown in \cite{EK}.  This was extended by Tardos \cite{ET} who established the existence of such an algorithm for ``combinatorial'' linear programs, that is, linear programs whose constraint matrix is 0-1 (or, more generally,  contains integers that are at most exponentially large in the dimensions).  However, no strongly polynomial algorithm is known for general linear programming.

The following summarizes one of the most important open problems in optimization and the theory of algorithms and complexity:

\begin{conjecture}
\label{con:Simplex_strongly_poly}
There is a strongly polynomial algorithm for linear programming.
\end{conjecture}

One particularly attractive direction for a positive answer for this conjecture is the search for polynomial variants of the simplex method.  It would be wonderful to discover a pivoting rule for the simplex method which (unlike all known such methods) always finds the optimum after a number of iterations that is polynomial in $d$ and $n$.   Hence the following is an interesting speculation:

\begin{conjecture}
\label{con:Simplex_strongly_poly_rand}
There is a pivoting rule for the simplex method that terminates after a number of iterations that is, in expectation, polynomial in $d$ and $n$.
\end{conjecture}

In relation to Conjecture \ref{con:Simplex_strongly_poly_rand}, clever randomized pivoting rules of a particular recursive sort were discovered rather
 recently, with worst-case number of iterations that has a subexponential dependence on $d$ \cite{GK,MSW}.
 On another front, the interesting polynomial simplex-like algorithm of Kelner and Spielman \cite{KS} does not settle Conjecture \ref{con:Simplex_strongly_poly_rand}
 because it is weakly polynomial, as the complexity of each iteration, and the number of iterations, depends (polynomially of course) on the bits of the integers in the input.
 Other recent results related to Conjecture \ref{con:Simplex_strongly_poly} can be found in \cite{Ch,YYY}.

In the next section we formalize the concept of a {\em pivoting rule:}  A method for jumping from one basic solution to an adjacent one that (1) is
strongly polynomial per iteration; (2) is guaranteed to increase a potential function at each step; and (3) is guaranteed to always terminate at
the optimum (or certify infeasibility or unboundedness).  We also give  several examples of such rules.   It is important to note that in our
definition we allow pivoting rules to jump to {\em infeasible bases} in order to include pivoting rules other than of the primal type.  Also, our original definition in Section 2 restricts pivoting rules to be deterministic; we discuss the important subject of randomized rules in Section 5.

Recently there has been a glimmer of hope that some stronger forms of the two conjectures
%Conjecture \ref{con:Simplex_strongly_poly}
 could be disproved, after the disproof of the Hirsch Conjecture \cite{SF}.  The Hirsch conjecture \cite{GBD} posited that the diameter of a $d$-dimensional polytope with $n$ facets is at most $n-d$, the largest known lower bound.  The best known upper bound for this diameter is the quasi-polynomial bounds of \cite{KK}.   But even a super-polynomial lower bound would only falsify
% Conjecture \ref{con:Simplex_strongly_poly}
 the conjectures for {\em primal} pivoting rules (ones going through only feasible bases, i.e., vertices of the polytope), but {\em not} for the many other kinds of pivoting rules (see the next section).  Furthermore, it is now clear that the techniques involved in the disproof of the Hirsch conjecture are incapable of establishing a nonlinear lower bound on the diameter of polytopes, and it is widely believed that there is a polynomial upper bound on the diameter of polytopes.

{\em In this paper we contemplate whether the concepts and methods of complexity theory can be applied productively to illuminate the problem of strongly polynomial algorithms for linear programming and Conjecture \ref{con:Simplex_strongly_poly}.}    We show a result suggesting that PSPACE-completeness may be relevant.

In particular, we propose to classify deterministic pivoting rules by the complexity of the following problem, which we call
 {\sc the path problem} of a pivoting rule:  Given a linear program and a basic  solution, will this latter one appear on the
  pivot rule's path? Recall that PSPACE is the class of problems solvable in polynomial {\em memory}.  This class contains NP, and it is
  strongly believed to contain it strictly.  The {\sc path problem} of a pivoting rule is clearly in PSPACE, because it can be solved by
following the (possibly exponentially long) path of  the rule, reusing space; if it is PSPACE-complete, then the pivoting rule cannot be polynomial (unless, of course, P = PSPACE).

But it is not a priori clear that  there are pivoting rules for which the path problem is PSPACE-complete.  We show (Theorem \ref{thm:Simple_is_PSPACE_hard}) that they do exist; unfortunately, we prove this not for one of the many classical pivoting rules, but for a new, explicitly constructed --- and fairly unnatural --- one.  We conjecture that the same result holds for essentially all known deterministic pivoting rules; such a proof looks quite challenging; obviously, in such a proof much more will need to be encoded in the linear program (which, in the present proof, is of logarithmic complexity and isomorphic to $\{0,1\}^n$).  However, we do exhibit (Theorem \ref{thm:shadow_is_tractable}) a pivoting rule whose path problem is in P:  It is Dantzig's well-known {\em self-dual simplex}  \cite{GBD} (also known  as {\em shadow vertex algorithm}), which is known to be exponential in the worst case \cite{KGM}, but has been used in several sophisticated algorithmic upper bounds for linear programming, such as average-case analysis and smoothness
\cite{HB,SS,AM,AKS,MT,ST}.   We briefly discuss the apparent connection between the average-case performance of a pivoting rule and the complexity of its path problem.

The motivation for our approach came from recent results establishing that it is PSPACE-complete to compute the final result of certain well known algorithms for finding fix points and equilibria \cite{GPS}.  However, the proof techniques used here are completely different from those in \cite{GPS}.

\section{Definitions}
Consider an algorithm whose input is an array of $n$ integers.  The algorithm is called {\em strongly polynomial} if
\begin{itemize}
\item it is polynomial-time as a Turing machine, and
\item if one assumes that all elementary arithmetic operations have cost one, the worst-case complexity of the algorithm is  bounded by a polynomial in $n$, and is therefore independent of the size of the input integers.
\end{itemize}

In linear programming one seeks to maximize $c^Tx$ subject to $Ax = b, x \geq 0$, where $A$ is $m \times n$.  An $m\times m$ nonsingular submatrix $B$ of $A$ is a
{\em basis}.  A {\em feasible basis} $B$ is one for which the system $Bx_B = b$ (where by $x_B$ we denote vector $x$ restricted to the coordinates
 that correspond to $B$) has a nonnegative solution; in this case, $x_B$ is called a {\em basic feasible solution.}  Basic feasible solutions
  are important because they render linear programming a combinatorial problem, in that the optimum, if it exists,  occurs at one of them.
   We say that two bases are {\em adjacent} if they differ in only one column.

There are many versions of linear programming (with inequality constraints, minimization, unrestricted in sign variables, etc.) but they are all known to be easily interreducible.  We shall feel free to express linear programs in the most convenient of these.

We shall assume that the linear programs under consideration are
non-degenerate (no two bases result in the same basic  solution).  Detecting this condition is nontrivial (NP-hard, as it turns out).
However, there are several reasons why this very convenient assumption is inconsequential.  First, a random perturbation of a linear program (obtained, say, by adding a random small vector to $b$) is non-degenerate with probability one.  And second, simplex-like algorithms can typically be modified to essentially perform (deterministic versions of) this perturbation on-line, thus dealing with degeneracy.

We next define a class of algorithms for linear programming that are variants of the simplex method, what we call {\em pivoting rules}.  To start, we recall from linear programming theory three important kinds of bases $B$, called {\em terminal bases}:

\begin{itemize}
\item  optimality: $B^{-1}b \geq 0, c^T -c_B^TB^{-1}A\leq 0$.  $B$ is the optimal  feasible basis of the linear program.
\item unboundedness: $B^{-1}A_j \leq 0, c_j-c_B^TB^{-1}A_j>0$ for some column $A_j$ of $A$.  This implies that the linear program is unbounded if feasible.
\item infeasibility: $(B^{-1})_iA\geq 0, (B^{-1})_ib<0$ for some row $(B^{-1})_i$ of $B^{-1}$.  This means the linear program is infeasible.
\end{itemize}
Notice that, given a basis, it can be decided in strongly polynomial time whether it is terminal (and of which kind).

\begin{definition} A {\em pivoting rule} $R$ is a strongly polynomial algorithm which, given a linear program $(A,b,c)$:
\begin{itemize}
\item produces an initial basis $B_0$;
\item given in addition a basis $B$ that is not terminal, it produces an adjacent basis $n_R(B)$ such that $\phi_R(n_R(B))>\phi_R(B)$, where $\phi_R$ is a potential function.
\end{itemize}

The {\em path} of pivoting rule $R$ for the linear program $(A,b,c)$ is the sequence of bases $(B_0, n_R(B_0), n^2_R(B_0),\ldots,, n_R^k(B_0))$, ending at a terminal basis, produced by $R$.
\end{definition}

Obviously, any pivoting rule constitutes a correct algorithm for linear programming, since it will terminate (by monotonicity and finiteness), and can only terminate at a terminal basis.  Notice that pivoting rules may pass through infeasible basic solutions (for example, they can start with one).  Incidentally, the inclusion of infeasible bases implies that such rules operate not on the linear program's polytope, but on its {\em linear arrangement}.  Since the latter has diameter $O(mn)$, even the existence of polytopes with super-polynomial diameter will not rule out strongly polynomial pivoting rules.

There are many known deterministic pivoting rules (ties are broken lexicographically, say):
\begin{enumerate}
\item  {\bf Dantzig's rule (steepest descent).}  In this rule (as well as in all other primal rules that follow), given a feasible basis $B$ we first calculate,
 for each index $j$ not in the basis the objective increase gradient $c_j^B=c_j-c^T_BB^{-1}A_j$.
    Define $J(B)=\{j: c^B_j>0\}$.  Dantzig's rule selects the $j\in J(B)$ with largest $c^B_j$ and brings it in the basis.
    By non-degeneracy (if not a terminal basis), this completely determines the next basis.  As with all primal pivoting rules, the potential function $\phi_R$ is the objective.

\item  {\bf Steepest edge rule.}  Instead of the maximum $c^B_j$, select the largest $c^B_j\over ||B^{-1}A_j||$.

\item  {\bf Greatest improvement rule.}  We bring in the index that results in the largest increment of the objective.

\item {\bf Bland's rule.}  Select the smallest $j\in J(B)$.

For all these rules, however, we have not specified the original basis $B_0$.  This is obviously a problem, since all these rules are primal and need feasible bases, and a feasible basis may not be a priori available.  Primal pivoting rules such as these are best applied not on the original $m\times n$ linear program $(A,b,c)$, but to a simple $m\times 2n$ variant called ``the big $M$ version,''  defined as $(A|-A), b, (c|-M,\ldots,-M)$, where $M$ is a large number ($M$ can either be handled symbolically, or be given an appropriate value computed in strongly polynomial time).   It is trivial now to find an initial feasible basis.  In fact, the pivoting rule running on the new linear program can be thought of as a slightly modified pivoting rule acting on the original linear program (when $j\in J(B)$, $A_j$ is negated, and $c_j$ is replaced by $-M$).

\item {\bf Shadow vertex rule.}  Here $B_0$ is any basis.  Given $B_0$, we construct two vectors $c_0$ and $b_0$ such that $B_0$ is a feasible basis, and also a dual feasible basis, of the relaxed linear program $\max c_0^Tx$ subject to $Ax = b_0, x\geq 0$.  Now consider the line segment between these two linear programs, with right-hand side and objective $\lambda b + (1-\lambda)b_0$ and $\lambda c + (1-\lambda)c_0$, respectively.  Moving on this line segment from $\lambda = 0$, we have both primal-feasible and dual-feasible (and hence optimal)
    solutions.  At some point, one of the two will become infeasible (and only one, by non-degeneracy).  We find a new basic solution by exchanging variables as dictated by the violation, and continue.  The potential function is the current $\lambda$.  When $\lambda =1$ we are at the optimum.

\item {\bf Criss-cross rules.}  A class of pivoting rules outside our framework, whose first variant appeared in \cite{SZ}, goes from one (possibly infeasible) basis to the other and convergence to a terminal basis is proved through a combinatorial argument that does not involve an explicit potential function.   However, certain such rules (such as the criss-cross pivoting rule suggested in \cite{TT}) have been shown (\cite{FM}) to possess a monotone potential function, and so they can be expressed within our framework.

\item {\bf Dual pivoting rules.}  Naturally, any of the primal pivoting rules can work on the dual.

\item {\bf Primal-dual pivoting rule.}  This classical algorithm \cite{DFF} is an important tool for developing simplex-inspired combinatorial algorithms for a broad set of network problems, acting as a reduction from weighted to unweighted combinatorial problems.  It does not conform to our framework, because it involves an inner loop solving a full-fledged linear program.

\item {\bf Pivoting rules with state.} Finally, also outside our framework are pivoting rules relying on data other than $A, b, c,$ and $B$, for example a pivoting rule relying on statistics of the history of pivoting such as selecting to include the index which has in the past been selected least frequently.

\item{\bf Randomized pivoting rules.}  There are several proposed randomized pivoting rules.  The ambition here is that the rule's expected path length is polynomial.  The simplest one \cite{GBD} is to pick a random index in $J(B)$.  Another important class of randomized rules are the {\em random facet} rules used in the proofs of subexponential diameter bounds \cite{GK,KK,MSW}.  We discuss randomized pivoting rules in Section 5.

\end{enumerate}

A pivoting rule is {\em strongly polynomial} if for any linear program  the length of the path is bounded above by a polynomial in $m$ and $n$. All pivoting rules within our framework mentioned above are known {\em not} to be strongly polynomial, in that for each one of them there is an explicit family of linear programs with non-polynomial path length, see \cite{AZ} for a unifying survey.

Explicit constructions are one way of ruling out pivoting rules.  {\em But is there a complexity-theoretic way? }   Our interest was sparked by the story of a well-known pivoting rule for a problem other than linear programming:  The Lemke-Howson algorithm for two-player Nash equilibrium, discovered in the 1960s \cite{LH}.  The first explicit construction was obtained decades later \cite{SV} and was extremely complicated.  More recently, it was established that the problem of finding the Nash equilibrium discovered by the Lemke-Howson algorithm is PSPACE-complete \cite{GPS} (and thus the algorithm cannot be polynomial, as long as P $\neq$ PSPACE).  Remarkably, the PSPACE-completeness proof was much simpler than the explicit construction.  We are led to the main definition of this paper:

\begin{definition}
The {\em path problem} associated with a pivoting rule $R$ is the following:  Given a linear program and a basis $B$, does $B$ appear on the path of $R$ for this linear program?

A pivoting rule is called {\em intractable} if its path problem is PSPACE-complete.  A pivoting rule is {\em tractable} if its path problem can be solved in strongly polynomial time.%\footnote{A possible criticism of this definition is that it seems to rule out pivoting rules of intermediate complexity.  If such a natural pivoting rule ever comes up, the definition will have to be revised...}
\end{definition}

The reason why this concept may be useful in understanding the complexity of linear programming is the following straightforward result:

\begin{proposition}
\label{prop:strongly_poly_intractable_PSPACE=P}
If an intractable pivoting rule is strongly polynomial, then PSPACE = P.
\end{proposition}

But are there examples of these two categories?  This is the subject of the next two sections.
\section{An Intractable Pivoting Rule}
This section is devoted to the proof of the following theorem.

\begin{theorem}
\label{thm:Simple_is_PSPACE_hard}
There is an intractable pivoting rule $R$.
\end{theorem}

The PSPACE-completeness reduction is based on the Klee-Minty construction, the original explicit exponential example for a variant of the
simplex method \cite{KM}, which we recall next.

The {\em $d$-dimensional Klee-Minty cube} is the following linear program:

\begin{eqnarray*}
&\max  x_1 \\
%\hbox{\rm subject to\ }
&0\leq  x_d \leq 1 \\
&\epsilon x_{i+1} \leq   x_i \leq 1-\epsilon x_{i+1}, i = 1,\ldots, d-1\\
&x_i\geq  0, i = 1,\ldots, d
\end{eqnarray*}

The feasible region of this linear program is a distorted $d$-hypercube (it obviously describes precisely the $d$-hypercube when $\epsilon = 0$):
 A polytope whose vertices are within a radius of $\epsilon$ from those of a hypercube, and are therefore in one-to-one correspondence
 with the elements of $\{0,1\}^d$.  Thus the feasible bases will also be represented as bit strings in $\{0,1\}^d$.  The objective function
 has a minimum at $0^d$ (a string of $d$ $0$'s) and a maximum at  $10^{d-1}$.

Let us now recall a well-known order on  $\{0,1\}^d$ called {\em Gray code} and denoted $G_d$.   $G_1$ is simply the order $(0, 1)$.    Inductively, the Gray code $G_{i+1}$ is $(0G_i, 1G_i^R)$, by which we mean, the sequence $G_i$ with each bit string preceded by a $0$, followed by the {\em reverse} of the order $G_i$, this time with each bit string preceded by $1$.  If $0\leq k < 2^d$, we denote by $G_d[k]$ the $k$-th bit string in $G_d$.

$G_d$ is a bijection between $\{0,1,\ldots,2^{d-1}\}$ and $\{0,1\}^d$, and therefore we can define the {\em successor function}
$S_d: \{0,1\}^d \mapsto \{0,1\}^d$ as follows: $S_d(x) = G_d[G_d^{-1}(x)+1]$.  The following is straightforward:

\begin{lemma}
$S_d$ can be computed in polynomial time.
\end{lemma}

Consider a vertex of the Klee-Minty cube of dimension $d$ --- equivalently, a bit string $(b_1,\ldots,b_d)\in \{0,1\}^d$.  This vertex has $d$ adjacent vertices, each obtained by flipping one of the $b_i$'s.  We call the $i$-th coordinate {\em increasing} at this vertex if the objective increases by flipping $b_i$.  The following are known important properties of the Klee-Minty cube:

\begin{lemma}
\begin{itemize}
\label{lem:klee-minty_order}
\item [(a)] The $i$-th coordinate is {increasing} if and only if $\sum_{j=1}^{i}b_j$ is even.
\item [(b)] Therefore the sequence of the vertices sorted in increasing objective is precisely $G_d$.
\end{itemize}
\end{lemma}

We next describe the starting PSPACE-complete problem.  Recall that a {\em Boolean circuit} is a directed acyclic graph whose nodes
are called {\em gates}, and are of several types:  {\em Input}, with in degree zero, {\em not}, with in degree one, and {\em or} and {\em and} with
in-degree two.  The sinks of this directed acyclic graph are called {\em outputs}.   Suppose that we are given a Boolean circuit $C$ with $n$
inputs and $n$ outputs.  If the inputs are set to a particular string $x\in\{0,1\}^n$, $C$ will ``compute'' an output
 string $C(x)\in\{0,1\}^n$.  Suppose that this particular circuit $C$ has the property that the Hamming ($L_1$, that is) distance
  between $x$ and $C(x)$ is always one (that is, $C$ just computes the index of the input to flip).  The {\em path} of $C$ is the sequence $(x_i, i = 0,\ldots)$,
 where $x_0=0^n$ and $x_{i+1}=C(x_i)$.  The $C$-{\sc path} problem is the following.

$C$-{\sc path}:  Given $C$ and $x_C \in \{0,1\}^n$, is $x_C$ on the path of $C$?  It is obviously in PSPACE (one need only try the first $2^n$ bit strings in the path of $C$, reusing space; if $x_C$ is not reached by that time, we are in a loop and $x_C$ will never be reached).

\begin{lemma}
There is a family of circuits $C$ of size polynomial in the number of inputs and of polynomial complexity such that $C$-{\sc path} is PSPACE-complete.
\end{lemma}

\paragraph{Sketch of proof:} We start from the halting problem for a linear bounded one-tape Turing machine with empty starting tape.
Without the single-bit restriction the circuit $C$ could just find the first transition that applies and implement
 it (the starting and ending strings can be made to be $0^n$ and $x_C$ respectively).  The single-bit restriction can be realized through a straightforward coding trick.

\medskip

The reduction proceeds as follows:  Given an input $x_C \in \{0,1\}^n$, we shall construct a linear program and a
basis $\hat B$ such that $\hat B$ lies on the path of rule $R$ (yet to be described) if and only if $x_C$ lies on the path of $C$.
   The linear program is the Klee-Minty cube of dimension $2n$.  The last (least significant) $n$ coordinates of the cube will serve
   to encode the current bit string on the path of $C$, while the first $n$ coordinates will maintain a counter in Gray code.
   We denote the last string of the Gray code, $10^{n-1}$, by $x_G$.
The sought basis $\hat B$ is taken to be $\hat B =x_Gx_C$.

Next we describe the pivoting rule $R$.   In fact, it suffices to define $R$ only on Klee-Minty cubes of even dimension --- on any other linear program, $R$ can be any pivoting rule, say steepest descent.  First, the initial basis of $R$ is $B_0=0^{2n}$.  Second, here is the description of how  $R$ modifies the current basis $B$ (which, since the linear program is the Klee-Minty cube of dimension $2n$ is represented by a bit string of length $2n$):

\medskip
{\bf Pivoting Rule $R$ on basis $B$}:
\begin{enumerate}
\item If $B=10^{2n-1}$, this is a terminal basis and we are done.  Otherwise, let $B = (B_1,B_2),$ each a string of length $n$.
\item If $B_2 = x_C$ then $R(B)=(S_n(B_1),B_2)$.
\item Otherwise, if $B_1 = x_G$ then $R(B)=(B_1,S_n(B_2))$.
\item Otherwise, construct the circuit with $n$ inputs in the family $C$.
\item Compute $C(B_2)$; suppose that $B_2$ and $C(B_2)$ differ in the $i$-th place (by assumption, they only differ in one).
\item If the $n+i$-th coordinate of B is increasing, then $R(B) = (B_1,C(B_2))$.
\item Otherwise, $R(B)=(S_n(B_1),B_2)$.
\end{enumerate}

To explain the workings of $R$, the first $n$ bits are a counter, and the last $n$ bits encode the current bit string on the path of $C$ from $0^n$.
If either the first $n$ bits are $x_G$ or the last $n$ bits are $x_C$, then $R$ just counts up in the other counter (Steps 2 and 3).  Otherwise, (Steps 4 and 5), $C(B_2)$ is computed.
The intention now is to update the last $n$ bits to be $C(B_2)$.  If the flipped coordinate happens to be increasing in $B$, then this is done
immediately (Step 6).  But if it is not, then we do the following maneuver:  We increment the $B_1$ counter by flipping the bit in $B_1$
that leads to the next string in the Gray code (Step 7).  This way,
in the
next invocation of the pivot rule the flipped bit {\em will} be increasing (by Lemma \ref{lem:klee-minty_order}(a)).

To show that $R$ is a pivoting rule, it remains to show that it is strongly polynomial, and that there is a potential function $\phi_R$ such that the pivot step of $R$ is always monotonically increasing.  The former is immediate.  For the latter, $\phi_R(B)$ is the value of the objective $x_1$ in the basic feasible solution represented by $B$.  It is easy to see by inspection of Steps 2, 3, 6, and 7 that in each of these four cases $\phi_R(R(B))>\phi_R(B)$.

Finally, we must show that $\hat B$ is on the path of $R$ if and only if $x_C$ is on the path of $C$.  If $x_C$ is on the
path of $C$ then eventually $B_2$ will be $x_C$, after at most $2^n-1$ steps,  and from then on Step 2 will be executed to
 increment the counter $B_1$.  This counter must go through $\hat B$ just before arriving at the terminal basis.
 % $10^{2n-1}$.
If $x_C$ is not on the path of $C$ then the path of $C$ will cycle until eventually Step 7 will be executed for a $2^n$-th time
(it can be easily checked that the cycling of the path of $C$ does not avoid Step 7), at which point $B_1=x_G$.  From then on $\hat B$ cannot be reached.
This completes the proof of Theorem \ref{thm:Simple_is_PSPACE_hard}.   \qed

\section{A Tractable Pivoting Rule}
The pivoting rule we proved intractable is not a natural one.  We conjecture that essentially all the pivoting rules described in the last section are intractable (even though proving such a result seems to us challenging).  However, here we point out that there {\em is} a natural, classical pivoting rule that is tractable:

\begin{theorem}
\label{thm:shadow_is_tractable}
The shadow vertex pivoting rule is tractable.
\end{theorem}
\begin{proof}
Given a linear program $(A,b,c)$, let $B_0$ be the initial basis, and let $b_0$ and $c_0$ be the corresponding initial values of the primal and dual
 right-hand-side vectors.  Given a basis $B$, we claim that the following is a necessary and sufficient condition that $B$ lies on the path of shadow vertex:

\begin{quote}
There is a real number $\lambda\in [0,1]$ such that $(1-\lambda)B^{-1}b_0 +\lambda B^{-1}b\geq 0$ and
 $(1-\lambda)(c_0^T-(c_0)_B^TB^{-1}A) + \lambda(c^T-c_B^TB^{-1}A) \leq 0$.
\end{quote}

In proof, any basis on the path has a non-empty interval of $\lambda$'s for which these inequalities hold.  And if for a given basis $B$ this condition is satisfied, then the inequalities are satisfied for a subinterval of $[0,1]$.  If we assume, for contradiction, that $B$ is not on the path of shadow vertex, then we can run the shadow vertex pivoting rule forward and backward  from $B$, and eventually arrive from a different path to the beginning and end, contradicting non-degeneracy.
As the condition is a system of $2n$ linear inequalities with one unknown, this completes the proof. \qed
\end{proof}

There is an interesting story here, connecting tractability of pivoting rules and the saga of the average-case analysis of the simplex method.  During the early 1980s, and in the wake of the ellipsoid algorithm, average analysis of the simplex method (under some reasonable distribution of linear programs) was an important and timely open question, and indeed there was a flurry of work on that problem \cite{HB,SS,AM,AKS,MT}.  It was noticed early by researchers working on this problem that one obstacle in analyzing the average complexity of various versions of the simplex method was a complete inability to predict the path of pivoting rules --- that is, the apparent intractability of the path problem we are studying here.  And this makes sense:  If one cannot characterize well the circumstances under which a vertex will appear on the path, it is difficult to deduce the average performance of the algorithm by adding expectations over all vertices.  Once Borgwardt  \cite{HB} and Smale \cite{SS} had the idea of using the shadow vertex pivoting rule in this context, further progress ensued \cite{AM,AKS,MT}.

\section{Randomized Pivoting Rules}
Many pivoting rules are explicitly randomized, aiming at good expected performance.  Our definition can easily be extended to include randomization:  In the definition of a pivoting rule, $R(B)$ is not a single adjacent basis, but a {\em distribution} on the set of adjacent bases (naturally, this set is polynomially small).  Any basis $B'$ in the support of $R(B)$ must satisfy $\phi_R(B')>\phi_R(B)$.  Obviously, deterministic pivoting rules are a special case, and therefore Theorems \ref{thm:Simple_is_PSPACE_hard} and \ref{thm:shadow_is_tractable} trivially apply here too.

What is slightly nontrivial is to define what ``intractable'' means in this case.  That is, what is the ``path problem'' for a randomized pivoting rule $R$?  We believe that the right answer is the following ``promise'' problem:

\begin{definition}
Fix a polynomial $p$ and a function $f: Z^2 \mapsto [0,1-{1\over p(m,n)}]$.  The {\em $(f,p)$-path problem} associated with a randomized pivoting rule $R$ is the following:  Given an $m\times n$  linear program and a feasible basis $B$, distinguish between these two cases:  $B$ appears in the execution of $R$ with probability (a) at most $f(m,n)$; and (b)  at least $f(m,n)+{1\over p(m,n)}$.
\end{definition}

The analog of Proposition \ref{prop:strongly_poly_intractable_PSPACE=P} is now:

\begin{proposition}
If a randomized pivoting rule $R$ is strongly polynomial in expectation, then the $(f,p)$-path problem of rule $R$ is in BPP, for all $f$ and $p$.
\end{proposition}

Recall that BPP is the class of all problems that can be solved by randomized algorithms, possibly with a small probability or error, see Chapter 10 in \cite{Papa}.

\section{Discussion}
Pivoting rules constitute a rich and interesting class of algorithmic objects, and here we focused on one important attribute: whether or not the path problem of a pivoting rule is tractable.  We have pointed out that there is an intractable pivoting rule, whereas a well-known classical pivoting rule  is tractable.  The most important problem we are leaving open is to exhibit a natural intractable pivoting rule.  For example, establishing the following would be an important advance:

\begin{conjecture}
Steepest descent is intractable.
\end{conjecture}

This looks quite challenging.  Obviously, in such a proof much more will need to be encoded in the linear program (which, in the present proof, was of logarithmic complexity).  The ultimate goal is a generic intractability proof that works for a large class of pivoting rules, thus delimiting the possibilities for a strongly polynomial algorithm.  For example: Are all primal pivoting rules (the ones using only feasible bases) intractable?

There are pivoting rules beyond linear programming, usually associated with the linear complementarity problem (LCP, see \cite{Cot}).  They generally do not have a potential function, and termination is proved (when it is proved) by combinatorial arguments.  Lemke's algorithm is a well-known general pivoting rule, known to terminate with a solution (or with a certification that no solution exists) in several special cases.  It is known to be intractable in general \cite{GPS}, but it can be shown to be tractable when the matrix is positive definite.  We conjecture that it is intractable when the matrix is positive semidefinite.

\bibliographystyle{\alpha}

 \end{document}